\documentclass{collective-intelligence}

\usepackage{amsmath}
\usepackage{algorithm}
\usepackage{algorithmic}
\usepackage{multirow}

\usepackage{graphicx}
\usepackage{epsfig}

\begin{document}
\bibliographystyle{agsm}

\title{Thermodynamic Principles in Social Collaborations}%
%
%
%
%

\numberofauthors{4} 
%
\author{
%
%
\alignauthor
Huan-Kai Peng\\
       \affaddr{Carnegie Mellon University}\\
       \affaddr{Pittsburgh PA 15213}\\
       \email{pumbaapeng@cmu.edu}
\alignauthor
Ying Zhang and Peter Pirolli\\
       \affaddr{Palo Alto Research Center}\\
       \affaddr{Palo Alto CA 94304}\\
       \email{\{ying.zhang,peter.pirolli\}@parc.com}
\alignauthor Tad Hogg\\
       \affaddr{Inst. for Molecular Manufacturing}\\
       \affaddr{Palo Alto, CA 94301}\\
       \email{tadhogg@yahoo.com}
}

\maketitle

\begin{abstract}
A thermodynamic framework is presented to characterize the evolution of efficiency, order, and quality in social content production systems, and this framework is applied to the analysis of Wikipedia. Contributing editors are characterized by their (creative) energy levels in terms of number of edits. We develop a definition of entropy that can be used to analyze the efficiency of the system as a whole, and relate it to the evolution of power-law distributions and a metric of quality. The concept is applied to the analysis of eight years of Wikipedia editing data and results show that (1) Wikipedia has become more efficient during its evolution and (2) the entropy-based efficiency metric has high correlation with observed readership of Wikipedia pages.
\end{abstract}

\section{I. Introduction}
\label{sec:1}
A social production system \cite{Benkler_2006} is characterized by a process in which the creative energies of large number of people contribute to large projects, mainly without traditional, centralized, hierarchical organizational mechanisms. One question that may be asked is whether such systems collectively adapt to more efficiently and effectively harness the creative energies of contributors and produce higher quality outputs. In this paper, we present a thermodynamic framework that allows us to characterize the efficiency of social production systems in several ways. We analyze the evolution of efficiency in Wikipedia and show how this relates to content quality in terms of readership demand.

In this paper, we consider the social production of content in Wikipedia as an open thermodynamic system. Each editor corresponds to a particle and the number of edits associates with a level of particle energy. We then exploit the concept of entropy in statistical mechanics to understand the underlying principles for efficiency in social production systems. We apply these insights to analyze the evolution of efficiency and quality in Wikipedia. 

We also explain the observed {\em power-law} distribution of editing activity as a consequence of how edits relate to {``}energy{''} devoted to Wikipedia in a thermodynamic setting.
At an aggregate level, the power-law distribution of user activity on peer-production websites could arise through two mechanisms. First, the observed
diversity of behavior could reflect an extreme heterogeneity of preference among
the potential user population (Wikipedia editors are ``born") \cite{Panciera_2009}.
Second, the behavior could be due to diversity of experience of new users after
they start participating on the site, e.g., positive or negative feedback from
other users on their contributions (Wikipedia editors are ``made") \cite{Wilkinson_2008}\cite{Ren_2007}. 
In contrast to these studies, we consider a third
possibility, namely that the diversity arises from a relatively short-term decrease
in effort required to make additional contributions with experience of editing. 
This corresponds to the general improvement people have with
cognitive tasks \cite{Newell_1981}\cite{Pirolli_1985}\cite{Shrager_1988}. Specifically, we propose that the power-law behaviors at the level of contributions arise largely due to the decreasing effort required for a given
user to make additional edits in a relatively short period of time (e.g.,
one month) or to a particular page. This leads to a {\em logarithmic} energy model for edits.
 
With each number of edits $v$ we associate ``energy" $\log(v)$. The base of the
logarithm is arbitrary, and we use the natural logarithm, which is the common
convention in statistical mechanics. 
This logarithmic dependence means, for instance, it is easier for someone to make their 10th edit, after making 9 already, than it is to make their 2nd edit. 
This is reasonable --- people gain experience with the subject matter of the page so presumably can contribute
an edit with less time. 
The fact that people who edit a lot return more often in monthly statistics of Wikipedia supports the logarithmic energy assumption. Observations that user activity rates are lognormally distributed \cite{Hogg_2009} \cite{Hogg_2010}, and multiplicative processes result in power-law distributions \cite{Mitzenmacher_2004} provide further evidence. 

We put social production into the context of statistical mechanics perspective, by defining the notions of entropy, energy, temperature and free energy, and deriving their relationships in power-law distributions. With the Wikipedia editing statistics, we show that entropy efficiency (entropy per energy) and entropy reduction (relative entropy w.r.t. its maximum) are good metrics for quantifying efficiency and quality of the collaboration.

The rest of this paper is organized as follows. Section II defines the entropy related metrics and discusses their relationship with power-law distribution. Section III looks at monthly editing data in eight years of Wikipedia and illustrates the evolution of entropy-based metrics. Section IV analyzes page-wise entropy metrics and their relationship to quality of production. With a set of readership data, we show high correlation between entropy efficiency and readership for Wikipedia pages. 


\section{II. Metrics for Order and Efficiency}
\label{sec:2}
We model social collaboration as an open thermodynamic system consisting of a set of particles, each holding a certain level of energy. We define a set of metrics analogous to thermodynamic quantities: energy, entropy and temperature, to measure the order and efficiency of such a system. We show that under a logarithmic energy level, a system would self-organize with power-law distributions under thermodynamic principles.
\subsection{Entropy, Energy and Temperature}
Let $I$ be a set of individuals, e.g., a set of particles or editors, and $V$ be a set of positive values that an individual can hold, e.g., energy of a particle or the number of edits. A {\em collection} is a mapping $I \rightarrow V$. For the Wikipedia editing system, $I$ is a set of editors and $V$ is a set of positive integers,  and $v_i$ is the number of edits editor $i$ contributed. 
Here a collection can be as large as the whole collaborative community, say, the whole Wikipedia, or as small as a sub-community, say, a page.

Let $s_v$ be the number of individuals in $I$ with value $v$ in $V$,  and $N=|I|$ be the total number of individuals in the collection, we define {\em entropy} for such a collection as:
\begin{equation}
S =-\sum_v p_{v}\log(p_v) \mbox{ where } p_v= \frac{s_v}{N}.
\label{eq:ent}
\end{equation}
In this definition, if all individuals have the same value, $S$ is minimized to be $0$, If, on the other hand, all individuals have different values, $S$ would be maximized to $\log(N)$.
The collection is in high \emph{order} if the entropy is low, or contributions are even among individuals, and in disorder if the entropy is high, or there is divergent contributions among individuals. A complex and effective system would be in a state between order and disorder \cite{Mitchell_1993}.

The entropy $S$ can also be defined for a collection with infinite number of individuals; it has a physical meaning and is related to energy and temperature in thermodynamics. When a thermodynamic system has a large number of independent particles, according to the Boltzmann distribution \cite{Ma_1985},
\begin{equation}
p_{u} \propto e^{-\frac{u}{kT}}
\end{equation}
where $p_{u}$ is the probability of a particle at a given state with energy level $u$, $k$ is Boltzmann constant, and $T$ is temperature.  The Boltzmann distribution ensures that the particles have exponential distribution with respect to the energy levels, i.e., high energy particles in a given state are much less likely than lower energy particles. Entropy $S$ for this distribution becomes:
\begin{equation}
S = - \sum_{u} p_{u} \log (p_u) = \frac{E}{kT}+\log(Z)
\label{eq:tS}
\end{equation}
where ${E} = \sum_u p_{u}u$ is the {\em average} energy per particle,
\begin{equation}
Z=\sum_u e^{-\frac{u}{kT}}
\label{eq:partition}
\end{equation}
is the partition function, and
$S$ corresponds to thermodynamic entropy {\em per particle}. Note that Eq. \ref{eq:tS} can be rewritten as 
\begin{equation}
{E} = kTS - kT\log(Z)
\end{equation}
and 
\begin{equation}
A={E}-kTS = -kT\log(Z)
\label{eq:A}
\end{equation}
is called {\em free energy} \cite{Goodstein_1975}. Free energy is an important concept in thermodynamics, which is the amount of useful energy that can transfer to work.

Now assuming that the energy for the number of edits is proportional to its logarithmic value (as argued in Section I), i.e., $u=\log(v)$, and that the Boltzmann distribution holds for logical particles, we have:
\begin{equation}
p_{v} \propto e^{-\frac{\log(v)}{kT}} \rightarrow p_{v} \propto v^{-\alpha}, \mbox{where} \hspace{1mm} \alpha = \frac{1}{kT}
\end{equation}
which is a power-law distributions in values in $V$, with power-law co-efficient $\frac{1}{kT}$. In other words, $\alpha$ corresponds to the inverse of temperature. The higher the $\alpha$, the lower the temperature. 

Power-law distributions have been discovered in many social media including Wikipedia \cite{Wilkinson_2008}. From thermodynamic principles, we argue that they come naturally because the energy level is logarithmic in terms of the levels of the activities.

\subsection{Entropy Reduction and Entropy Efficiency}
Entropy measures disorder or the amount of uncertainty in the system. In contrast to entropy, {\em entropy reduction} (i.e., {\em effective entropy} \cite{Tononi_2008}) is a relative entropy with respect to the maximum entropy given the number of individuals in the collection:
\begin{equation}
R = \log(N) - S
\end{equation}
where $N$ is the total number of individuals and $S$ is the entropy defined in Eq. \ref{eq:ent}.
$R$ is maximum when $S$ is minimum --- $S$ measures disorder; $R$ measures order, and in particular, the amount of entropy reduction due to order. If the collection is from a random uniform distribution, $S$ will be close to $\log(N)$ and $R$ close to 0.

Let ${E}$ be the average energy defined in Eq. \ref{eq:tS}, we have {\em entropy efficiency} as average entropy per unit energy level:
\begin{equation}
Q = \frac{S}{E}.
\label{eq:eff}
\end{equation}
Note that according to Eq. \ref{eq:tS}, 
\begin{equation}
Q = \frac{S}{E} = \frac{1}{kT} + \frac{\log(Z)}{E} = \alpha + \frac{\log(Z)}{E}.
\label{eq:Q}
\end{equation}
The following theorem claims that power-law distributions maximize entropy efficiency for logarithmic energy levels.
\newtheorem{theorem}{Theorem}
\begin{theorem}
Power-law distributions maximize entropy efficiency $Q$ (Eq. \ref{eq:eff}) when ${E} = \sum_v p_v \log(v)$.
\label{th:eff}
\end{theorem}
\begin{proof}
The proof is from \cite{Mitzenmacher_2004}. To maximize Eq. \ref{eq:eff}, let the derivative of $Q$ with respect to $p_v$ to be 0. The derivative of $Q$ with respect to $p_v$ is, 
\[
(\frac{\partial S}{\partial p_{v}}{E} - \frac{\partial {E}}{\partial p_{v}} S)/{E}^2 
\]
i.e., performing derivative of $S$ w.r.t. $p_v$ in Eq. \ref{eq:ent} and derivative of logarithmic energy $E$ w.r.t. $p_v$, we have
\[
[-(\log(p_v)+1){E}-\log(v)S]/{E}^2.
\]
When it is 0, we have
\begin{equation}
\log(p_v)+1=-\log(v)\frac{S}{E}=\log(v^{-\alpha})
\label{eq:effprove}
\end{equation}
with $\alpha=\frac{S}{E}=Q$, i.e., $p_v \propto v^{-\alpha}$.	
\end{proof}
Theorem \ref{th:eff} indicates not only that the power-law maximizes the entropy efficiency $Q$, the power-law coefficient is \emph{approximated} to it. 

Note that entropy efficiency can also be defined for linear energy levels with ${E} = \sum_v p_v v$. In this case, maximizing $Q$ (Eq. \ref{eq:effprove}) gives
\begin{equation}
\log(p_v)+1=-v\frac{S}{E}=-\alpha v
\end{equation}
i.e., $p_v \propto e^{-\alpha v}$, an exponential distribution as the Boltzmann distribution. Therefore, entropy efficiency is maximized when the corresponding distribution matches the energy levels: power-law for logarithmic and exponential for linear. Analogous to entropy maximization in the second law of thermodynamics for closed systems, entropy efficiency is another general principle that a system would use the minimum energy to produce the same amount of entropy.

The relationship between entropy efficiency $Q$ (shown in Eq. \ref{eq:Q}) and free energy $A$ (shown in Eq. \ref{eq:A}) is obtained by combining these two equations:
\begin{equation}
\frac{Q}{\alpha} = \frac{{E} - A}{E} .
\end{equation}
We call $\frac{Q}{\alpha}$ the {\em free energy reduction ratio}.

\subsection{Entropy Related Properties in Power-law}
Power-law distributions are dominant in the Wikipedia editing statistics where their properties can be fully captured by the power-law coefficient $\alpha$. In the rest of this section, we show how the change of $\alpha$ affects entropy $S$, entropy efficiency $Q$ and entropy reduction $R$, as well as energy per particle ${E}$ and free energy per particle $A$.

Figure \ref{fig:ent_syn} shows the change of entropy, entropy efficiency and entropy reduction, for two sample sizes, 1000 and $10^7$, with a varying power-law coefficient $\alpha$, where dashed lines are for entropy $S$, entropy efficiency $Q$ or entropy reduction $R$ in the uniform distribution with the corresponding number of samples. 
\begin{figure}
	\centering
	\includegraphics[width=2.8in]{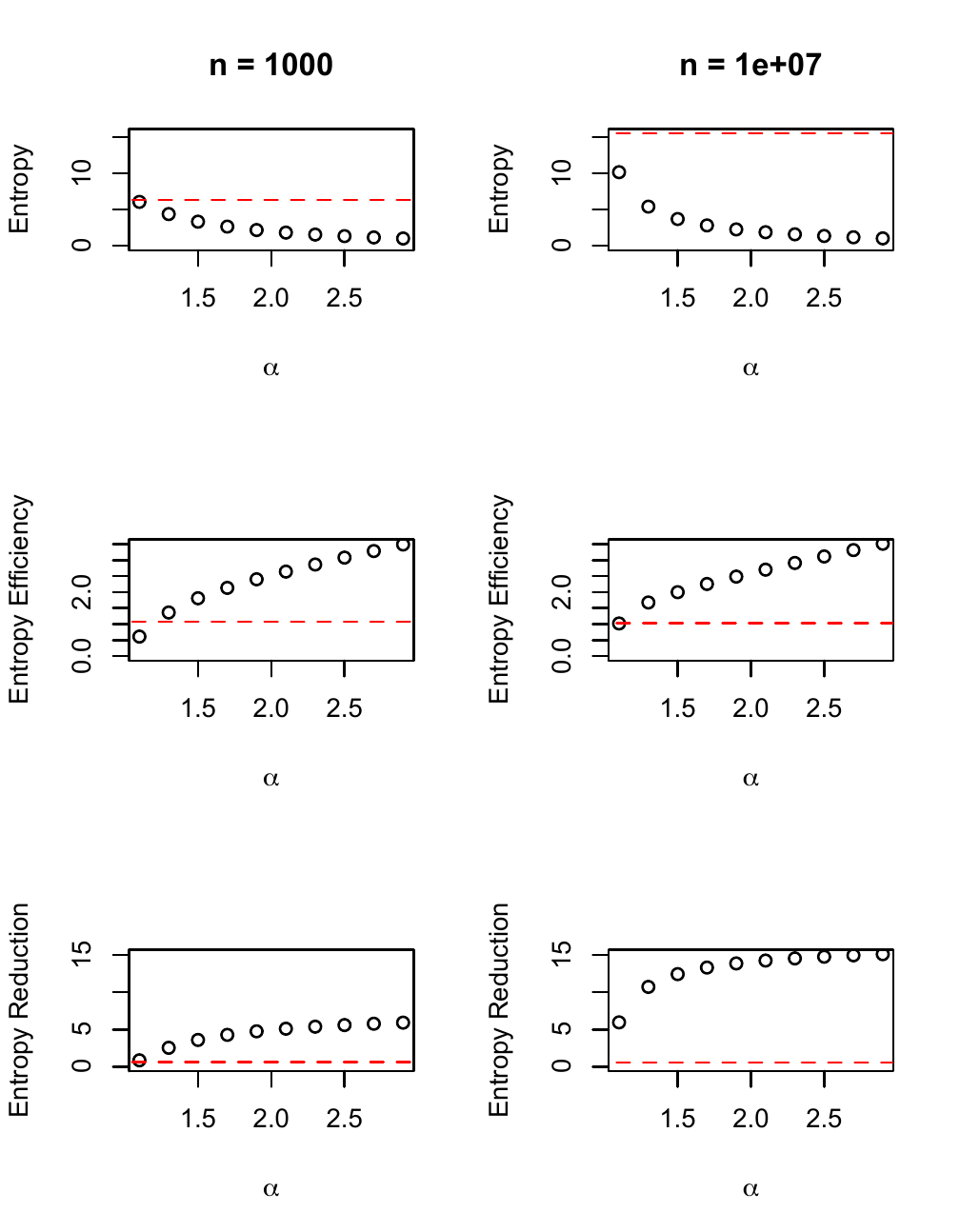}
	\vspace{-0.2in}
	\caption{Entropy, entropy efficiency and entropy reduction vs. power-law coefficients, in two sample sizes: $10^3$ and $10^7$.}
    \label{fig:ent_syn}
    \vspace{-0.0cm}
\end{figure}
We observe that:
\begin{itemize}
\vspace{-0.1in}
\item Both entropy and entropy reduction are growing with the number of samples. However, entropy efficiency is independent of the number of samples. Entropy efficiency grows almost linearly with the power-law coefficient.
\vspace{-0.1in}
\item Uniform distributions have maximum entropy, minimum entropy efficiency and minimum entropy reduction. Clearly, random uniform distributions do not have order or efficiency.
\vspace{-0.1in}
\item Entropy decreases with and entropy reduction increases with the power-law coefficient. However, the rates of increase and decrease reduce with the increase of $\alpha$, and both entropy and entropy reduction are saturated when $\alpha$ approach 2.5. 
\end{itemize}

For the Wikipedia editing statistics, the minimum value is one edit. 
For power-law distributions with minimum value 1, the relationship between the average energy and the power-law coefficient becomes:
\begin{equation}
{E} =\frac {1}{\alpha - 1}
\label{eq:eng}
\end{equation}
where ${E}$ is the average energy per particle. The proof of this relationship is in the Appendix. This implies that (1)  $\alpha>1$ (or $kT<1$) , and (2) $\alpha \rightarrow 1$, ${E} \rightarrow \infty$. Also if the  minimum $v$ to be 1, we have
\begin{equation}
Z  = \sum_{1}^{\infty} v^{-\alpha} = \zeta(\alpha)
 \end{equation}
 where $\zeta()$ is zeta function. Therefore the free energy (Eq. \ref{eq:A}) for a power-law distribution with coefficient $\alpha$ is
 \begin{equation}
 A = -kT\log(Z) = - \frac{\log(Z)}{\alpha} = -\frac{\log(\zeta(\alpha))}{\alpha}.
 \label{eq:freeeng}
 \end{equation}
Figure \ref{fig:energy} shows the relationship between energy (Eq. \ref{eq:eng}) and free energy (Eq. \ref{eq:freeeng}) with a varying power-law co-efficient $\alpha$. We see that when $\alpha$ increases, or temperature decreases,  average energy per particle decreases and free energy increases, both saturated when $\alpha$ approches 4, or $kT$ around 0.25.
\begin{figure}
	\centering
	\includegraphics[width=3.0in]{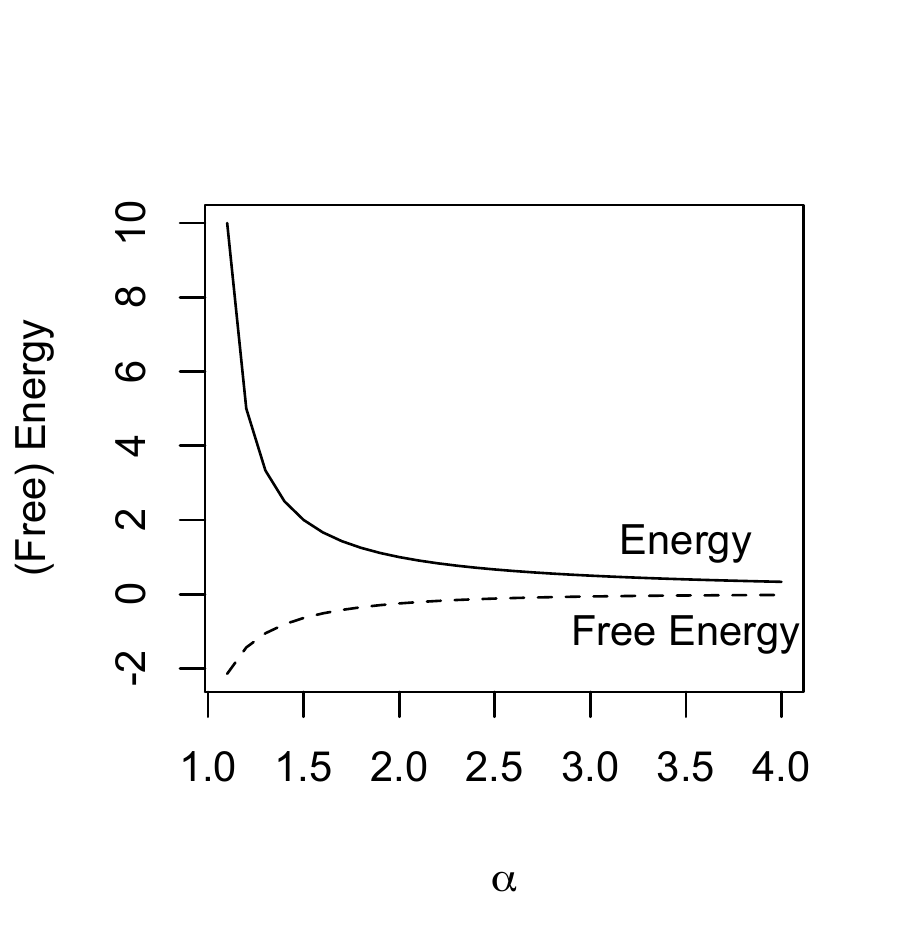}
	  \vspace{-0.2cm} 
	\caption{Energy and free energy with a varying power-law coefficient $\alpha$.}
    \label{fig:energy}
    \vspace{-0.2cm} 
\end{figure}

From these observations, we come to an explanation why in the real world, the power-law coefficient will lie mostly within $1<\alpha<3$. When $\alpha>3$, both free energy and entropy reduction are saturated, meaning it is not very useful for $\alpha$ to be larger or to have lower temperature.

Power-law distributions show fractal structures at multiple levels, in particular, if particles are grouped using their logarithmic scales, the resulted distribution is still power-law. Let $c_n$ be a class of units whose values are between $b^n$ to $b^{n+1}$, where $b>0$ is an arbitrary base. The total number of particles for class $n$ is $N(n)$ that is proportional to:
\begin{equation}
\int_{b^n}^{b^{n+1}} v^{-\alpha} dv = \frac{b^{-n(\alpha-1)}}{\alpha-1}(1-b^{-(\alpha-1)}) \propto b^{-n(\alpha-1)}
\label{eq:class1}
\end{equation}
and the total amount of values from class $n$ is $C(n)$ that is proportional to:
\begin{equation}
\int_{b^n}^{b^{n+1}} v\cdot v^{-\alpha} dv = \frac{b^{-n(\alpha-2)}}{\alpha-2}(1-b^{-(\alpha-2)}) \propto b^{-n(\alpha-2)}.
\label{eq:class2}
\end{equation}
Note that when $\alpha <2$, $C(n)$ increases exponentially with $n$, i.e., more portions of energy come from high energy  particles and when $\alpha >2$, $C(n)$ decreases exponentially with $n$, i.e., more portions of energy come from low energy class particles. When $\alpha=2$, $C(n)$ is a constant, i.e., all classes contribute evenly.

In conclusion, when $\alpha$ increases, entropy decreases, entropy reduction and entropy efficiency increase, and overall contributions are shifted from high energy particles to low energy particles.

\section{III. Is Wikipedia becoming more efficient?}  
We analyzed the Wikipedia editing data from January 2002 to December 2009. Prior research has shown that the growth of Wikipedia follows Logistic or Gomperz \cite{Suh_2009} curves and power-law distributions are everywhere \cite{Wilkinson_2008}. We would like to find out what properties other than volume growth have changed during the evolution of Wikipedia. In particular, we would like to answer the question of the degree to which Wikipedia becomes more efficient, from thermodynamic principles.

We consider Wikipedia as an open system with active editors in each month as particles, and their total number of edits of the month as values. The fact that the distributions of edits are power-law suggests that energy is logarithmic in terms of the number of edits. The system is open in the sense that there are editors joining and dropping from month to month. The sum of the logarithmic contributions from active editors (with minimum one edit) is the total energy of the system for the month. The number of active editors ranges from 1000 in early months to 600,000 in later months. We will overview the evolution of entropy, entropy reduction, and entropy efficiency of the Wikipedia's monthly editing activities.

\begin{figure}
	\centering
	\includegraphics[width=2.8in]{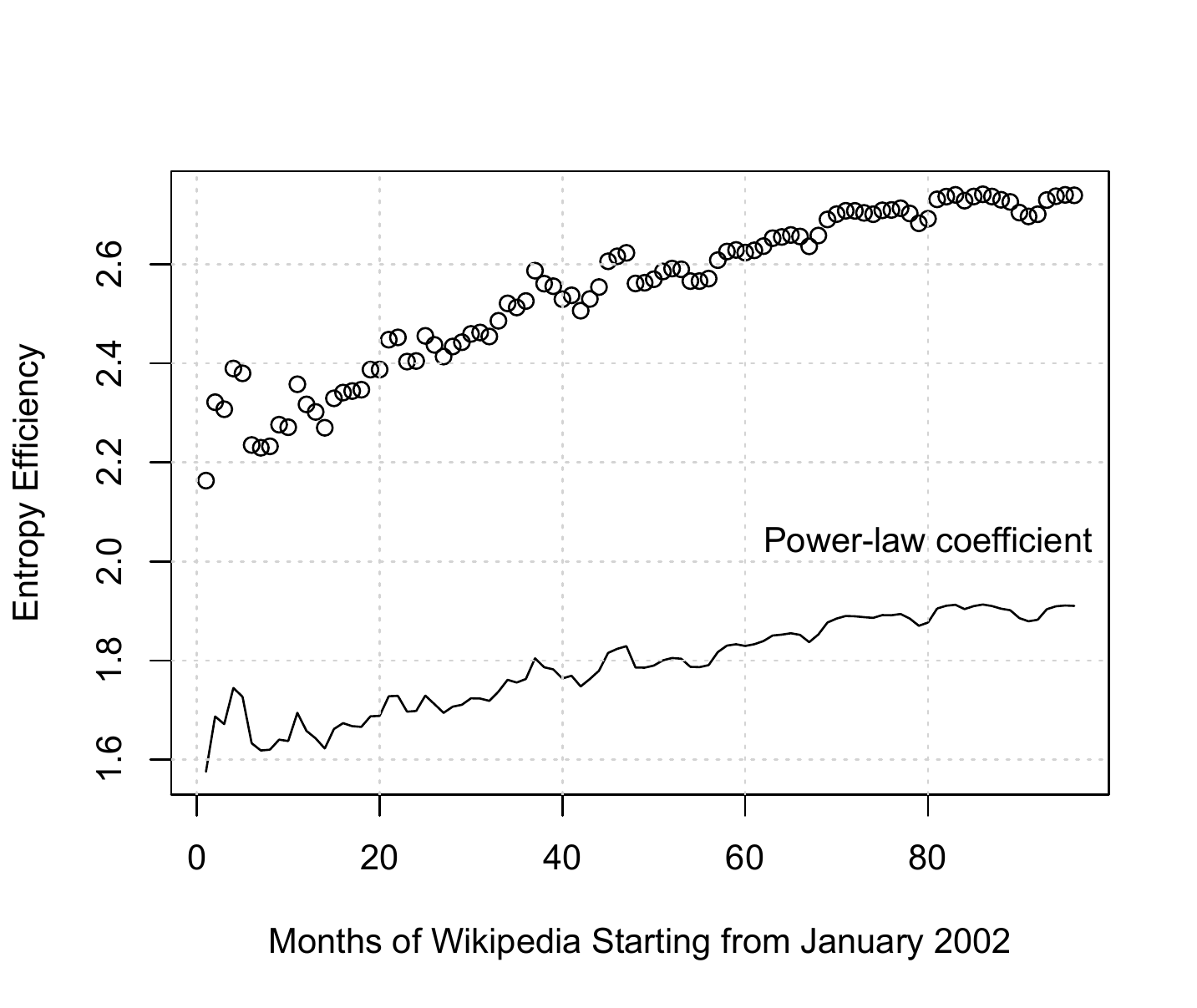}\\
	(a) \\
	\includegraphics[width=2.8in]{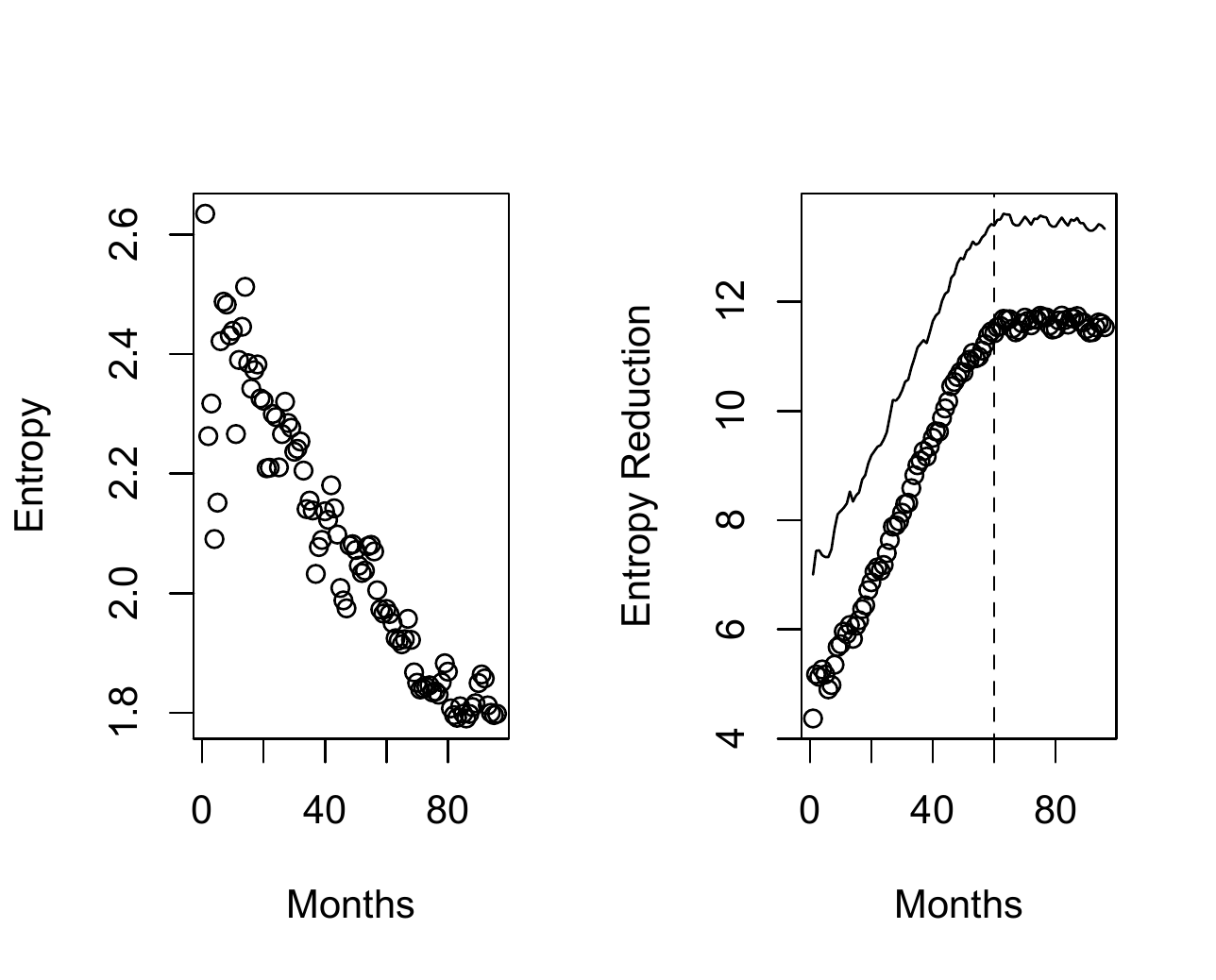}\\
	(b) \\	
	\includegraphics[width=2.8in]{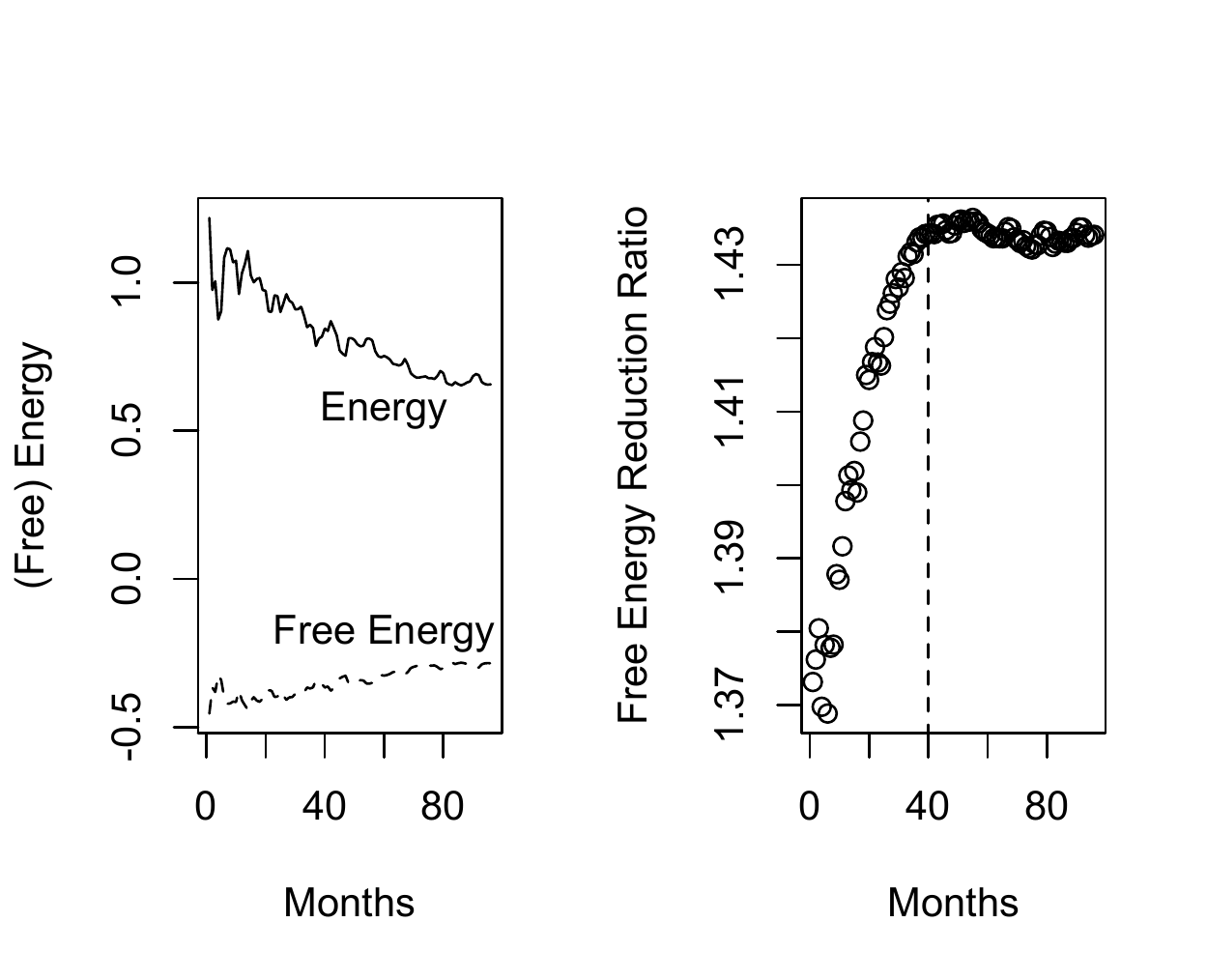}\\
	(c)	
	\caption{(a) Evolution of entropy efficiency and power-law coefficient in Wikipedia.
	(b) Evolution of entropy and entropy reduction in Wikipedia over 96 months. The solid line indicates the maximum entropy, $\log(N)$, where $N$ is the number of editors of the month.
	(c) Evolution of energy and free energy per editor, and free energy reduction ratio.
	}
    \label{fig:TEntropy}
\end{figure}
Figure \ref{fig:TEntropy}(a) shows entropy efficiency and power-law coefficient over the 96 months of the evolution. As we observe, power-law coefficient has grown from 1.5 to 2.0 steadily over the months and entropy efficiency has grown in almost the same rate. 
Figure \ref{fig:TEntropy}(b) shows the evolution of entropy and entropy reduction over 96 months of the history. Here the decreasing of entropy and increasing of entropy reduction suggest the increasing order in the editing system. 
Figure \ref{fig:TEntropy}(c) shows the change of energy per editor, the growth of free energy and the free energy reduction ratio. The ratio has been growing but saturated almost 20 months (at 40 months) before the saturation of the active editors (at 60 months). This seems to suggest that the saturation of the free energy reduction ratio maybe the cause of the saturation of the number of editors.

Another interesting question is -- how does the editor structure evolve? We classify editors according to their levels of energy, or logarithm of their edits, i.e., 1-10 edits as class 1, 11-100 edits as class 2, 101-1000 edits as class 3, etc. As we discussed in Section II, the number of editors in class $n$ is proportional to $10^{-n(\alpha-1)}$ (Eq. \ref{eq:class1}), and the total edits from class $n$ is proportional to $10^{-n(\alpha-2)}$ (Eq. \ref{eq:class2}). Since $\alpha$ has been increasing over the months, contributions are shifted from higher classes to lower classes, and now relatively even from all classes since $\alpha$ approaches 2 (from Eq. \ref{eq:class2}). It confirms that Wikipedia is becoming a media for the masses in later months, rather than for elites in early months. 
Figure \ref{fig:Editors_evo} and Figure \ref{fig:Edits_evo} show the total number of editors in each class and total contributions from each class, respectively, during 96 months of evolution. As we observe, except for high level classes which are noisy, the logarithmic volume of each class is proportional to the class index (Figure \ref{fig:Editors_evo}) and overall contributions from each class are getting relatively even (Figure \ref{fig:Edits_evo}), a result from increasing power-law coefficient in power-law distributions.
\begin{figure}
	\centering
	\includegraphics[width=3.0in]{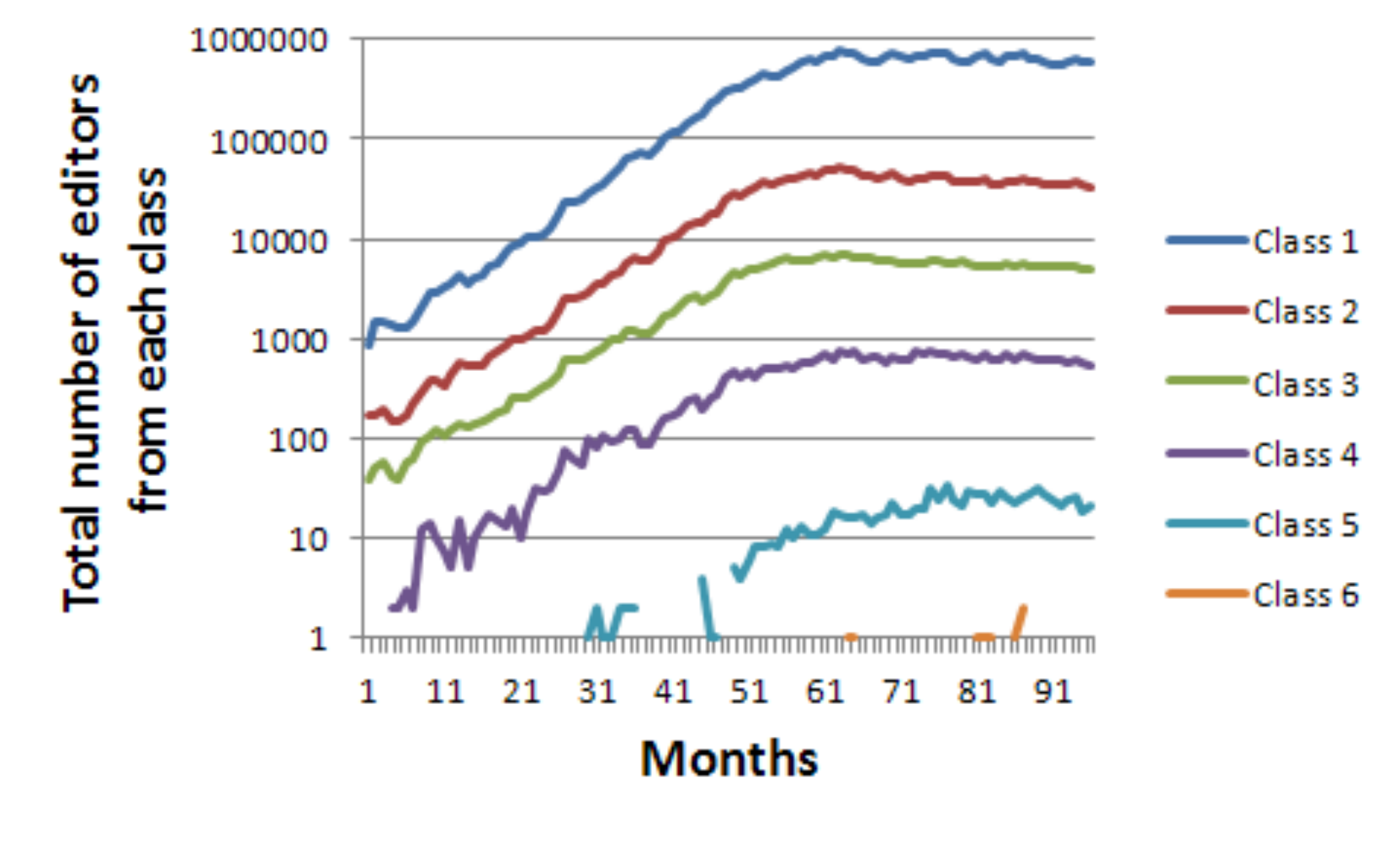}
	\caption{The total number of active editors from different classes over 96 months.}
    \label{fig:Editors_evo}
    \vspace{-0.0cm}
\end{figure}
\begin{figure}
	\centering
	\includegraphics[width=3.0in]{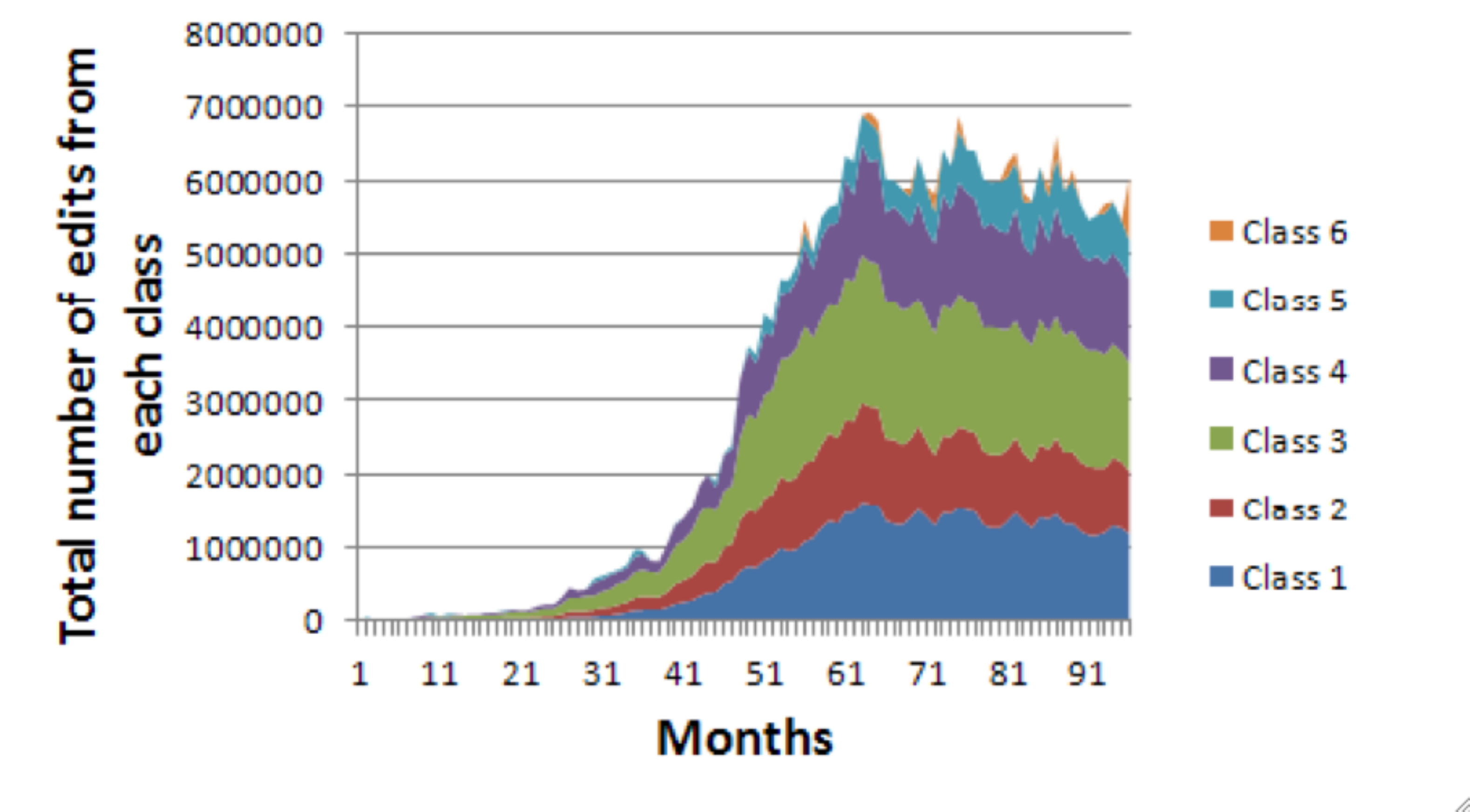}
	\caption{Total number of edits from each class over 96 months.}
    \label{fig:Edits_evo}
    \vspace{-0.0cm}
\end{figure}

Entropy-based metrics support previous findings about collaborative editing and bring new perspectives. In \cite{Wilkinson_2008}, the author reported that for a web media, the higher power-law coefficient $\alpha$, the higher entry barrier it has for an author to contribute a new edit. For example, it is easier to contribute a new edit to Digg (low $\alpha$), whereas it is harder to contribute a new edit to Wikipedia (high $\alpha$). We see that in Wikipedia, $\alpha$ has been growing steadily (Figure \ref{fig:TEntropy}(a)). Due to increasing order and efficiency, it is getting harder to add more edits for an editor. 

In conclusion, we believe Wikipedia has become more efficient in terms of entropy efficiency, and more ordered according to entropy reduction. The increasing power-law coefficient causes the shift of the contributions from elites to crowd. The saturation of free energy reduction ratio may cause the saturation of the active editors.

\section{IV. Does Efficiency Imply Quality?}
We have noticed that Wikipedia as a whole has evolved to be more efficient. We now like to see if those measurements are applicable to identifying the quality of pages. 
Based on this analysis, we examine (1) how the efficiency of a page correlates to its readership or quality, and (2) what is the major factor that separates high and low quality pages?

\subsection{Data Setting}
We choose the pages with at least 4500 edits and that are saturated as of December 2009. By saturation we mean that a page has gained less than 5\% edits in the latest 10\% of time since its creation. This selection ensures that each page being analyzed has gained a stable editing structure that is less noisy. Accordingly, there are a total of 962  pages being analyzed, and entropy-based measurements for each page are computed.

We have shown that power-law distributions maximize entropy efficiency for logarithmic energy levels. To distinguish \emph{power-law pages} from \emph{non-power-law pages}, we use the Komogrov-Smirnov statistic $D$ \cite{Casella_2001}, which is the maximum difference of the fitted and the empirical c.d.f.'s. Empirically, we found that $D = 0.1$ provides a good separation of power-law pages from those that are not. In total, there are 906 power-law pages and 56 non-power-law pages. 

\begin{figure*}
	\centering
	\includegraphics[width=7.0in]{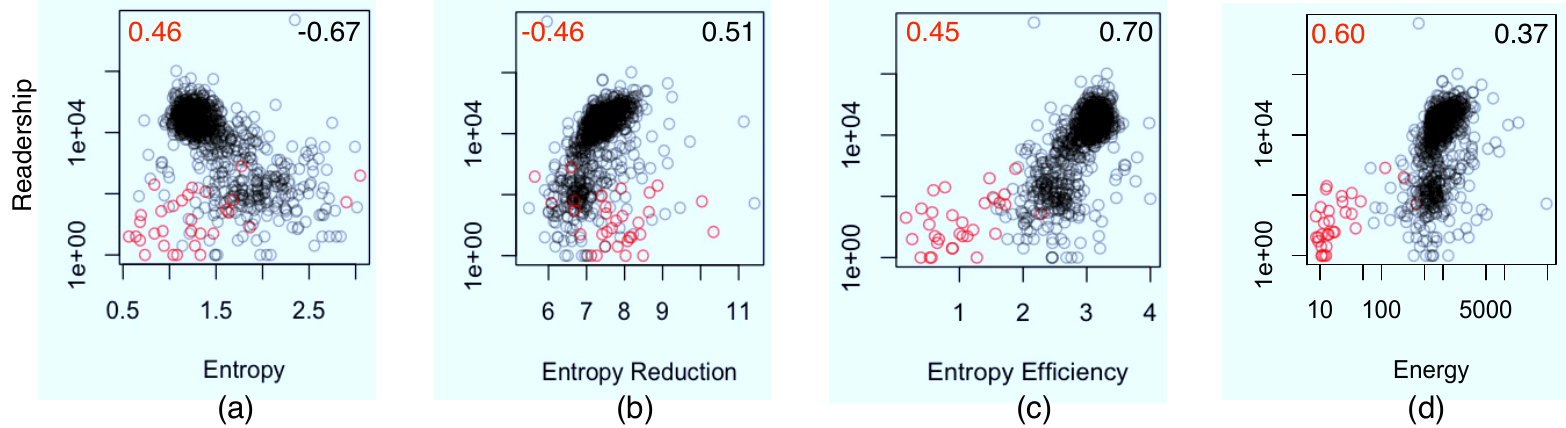} \\ \vspace{0.2cm}
	\includegraphics[width=7in]{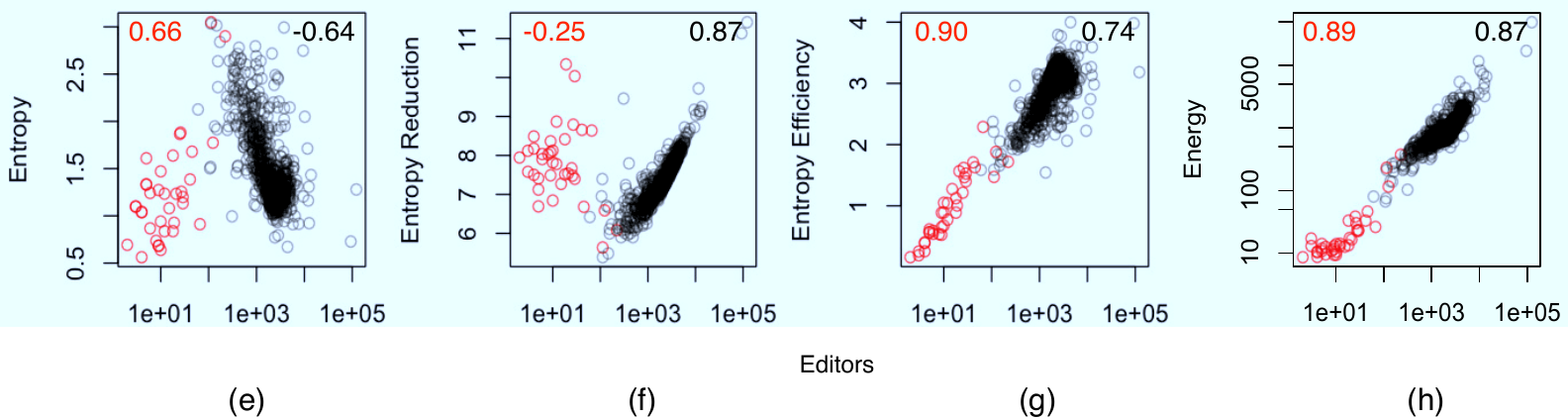} 
\caption{Top: The Readership of pages versus (a) Entropy, (b) Entropy Efficiency, (c) Entropy Reduction, and (d) Total Energy. Bottom: The number of editors of pages versus (e) Entropy (f) Entropy Efficiency (g) Entropy Reduction and (h) Total Energy. Black circles indicate power-law pages and red non-power-law.}
    \label{fig:reader_scatter}
    \vspace{-0.0cm}
\end{figure*}
\subsection{Efficiency v.s. Readership}
One of the main questions to answer is whether the order or efficiency of a page, measured by entropy-based metrics, corresponds to its quality. Here we look at a page's readership, the number of clicks it gets during some time interval. We use the readership data of one week in February 2009, but we noticed that using data of other points of times within half of the year also obtain similar behaviors. In Figures \ref{fig:reader_scatter}(a-d), we summarize how readership relates to (a) entropy, (b) entropy reduction, (c) entropy efficiency, and (d) total energy for each page, in which non-power-law pages are marked as red. The correlation coefficients ($\rho$) between readership and these metrics are also shown on the figures.

For the power-law pages (black circles), high readership associates with low entropy (Figure \ref{fig:reader_scatter}(a), $\rho$ = -0.67), high entropy reduction (Figure \ref{fig:reader_scatter}(b), $\rho$ = 0.51), and high entropy efficiency (Figure \ref{fig:reader_scatter}(c), $\rho$ = 0.70)---all suggesting that higher order and efficiency associate with higher readership. This is remarkable because there is no connection between readership and any of these metrics from how they are computed. One intuitive explanation is: as the editing structure become more ordered and efficient, the quality of the produced content improves, and thus the page draws more readers. 

Note also for power-law pages,  content quantity (in this case, the total energy) does not affect readership much, since there is only very small correlation ($\rho = 0.37$) between total energy and readership from Figure \ref{fig:reader_scatter}(d). Therefore we claim that quantity does not imply quality, but efficiency does.

On the other hand, all non-power-law pages have low readership: 
the maximum readership of non-power-law pages is 1111, whereas the median of all pages and power-low pages are 11163 and 11964, respectively. 
In addition, there is a positive correlation between entropy and readership (instead of negative for power-law cases) and negative correlation between entropy reduction and readership (instead of positive). The total energy of non-power-law pages are significantly lower. There is, however, positive correlations between total energy and readership for non-power-law cases ($\rho = 0.60$). The most interesting fact is: entropy efficiency, however, is correlated positively with readership in both power-law ($\rho=0.70$) and non-power-law ($\rho=0.45$) cases. 

In addition to total energy, we have also analyzed the correlation between total number of edits with readership (Figure \ref{fig:edits}). Note that there is almost no correlation between these two, for both power-law ($\rho=0.05$) and non-power-law cases ($\rho = 0.12$). In addition, although non-power-law pages are low in readership,  they in fact have more edits than the power-law ones in terms of mean ($9425.9 > 6537.4$) and median ($5532.5 > 4677.5$).  
The total energy, however, correctly separates power-law from non-power-law pages, i.e., low energy ones are not power-law.
This further reinforces our choice of using logarithmic energy in thermodynamic equations in this context.
\begin{figure}[h]
	\centering
	\includegraphics[width=3.2in]{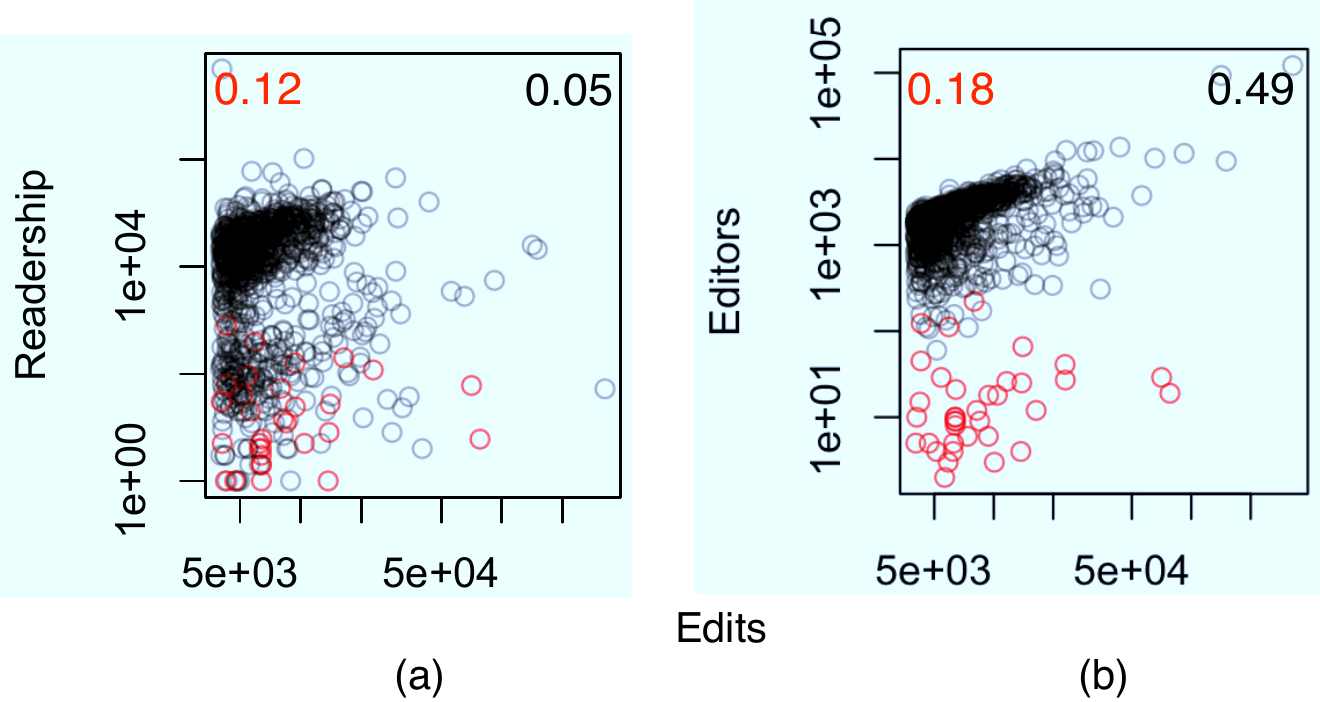}
	\caption{Total edits vs. readership and number of editors, where non-power-law pages are marked as red.}
    \label{fig:edits}
    \vspace{-0.0cm}
\end{figure}

\subsection{Efficiency vs. Editor Base}
In this section we like to show what actually correlates the most with efficiency and readership. The answer is, surprisingly, the number of editors.
Figures \ref{fig:reader_scatter}(e-h) summarize how the same metrics in Figures \ref{fig:reader_scatter}(a-d) correlate to a page's number of editors, where the correlation coefficients are shown on the figures.  First of all, for all non-power-law pages, the number of editors is small ($<100$).  In contrast to power-law pages, the number of editors is positively correlated to entropy ($\rho=0.66$) and slightly negatively correlated to entropy reduction ($\rho=-0.25$). However, what the most interesting fact is: the number of editors is highly positively correlated with entropy efficiency in both power-law ($\rho=0.74$) and non-power-law $(\rho=0.90)$ cases.  And, not too surprisingly, but still interestingly, the total energy is highly positively correlated with the number of editors, in both power low ($\rho=0.87$) and non-power-law ($\rho = 0.89$) cases. The simple fact is that the more editors the more efficiency, and the more efficiency the better quality, and these three metrics maybe positively reinforce each other.

From Figures \ref{fig:reader_scatter}(e-h), we see that the separation between clusters become clearer for all metrics. Roughly 100 editors seems to be the boundary between power-law and non-power-law pages: before this boundary (non-power-law pages), the increase on editors results in lower order, characterized by increased entropy and decreased entropy reduction. After this boundary (power-law pages), more editors results in higher order, i.e., lower entropy ($\rho$ = -0.64) and higher entropy reduction ($\rho$ = 0.87). This suggests that the system may have a phase transition at the boundary - growing from order to disorder and then from disorder to order. In both cases, however, the entropy efficiency and total energy grow with the number of editors.

Another interesting distinction introduced by the transition is \emph{elitism} versus \emph{wisdom of the crowd}. From Figure \ref{fig:edits} (b), we see that non-power-law pages show a form of elitism, characterized by relatively few elite editors single-handedly contribute an un-proportional amount of edits. In contrast, the power-law pages show a form of crowd wisdom, characterized by much more editors coming up with comparable, or slightly smaller number of edits. Since Figure \ref{fig:reader_scatter} already shows that power-law pages tend to have more readership than non-power-law ones, it implies that the nature of Wikipedia is a true media of the masses, where pages produced by crowd wisdom will have higher quality and thus more readership compared to that produced by a few elites.

We have shown that high power-law coefficient (or low temperature) implies high entropy efficiency. It is interesting that it is also true for non-power-law cases in our data.
\begin{figure}
	\centering
	\includegraphics[width=1.6in]{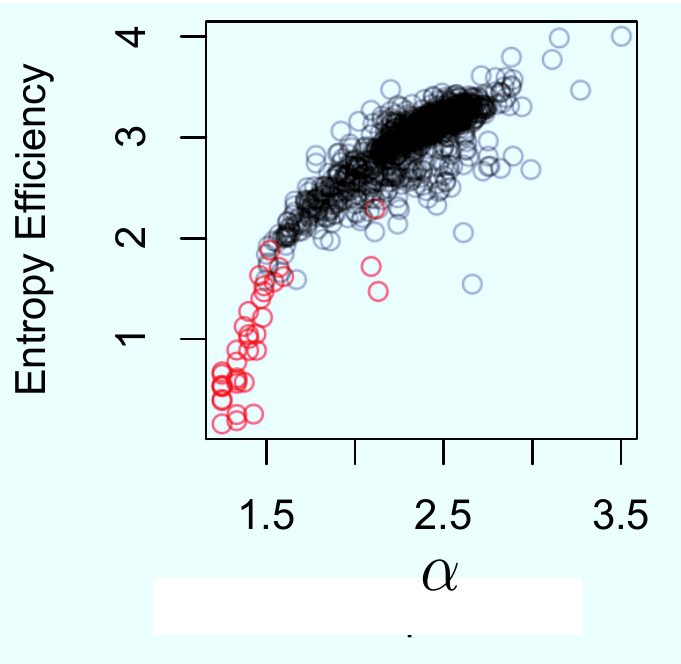}
    \vspace{-0.7cm}
	\caption{The power-law slope $\alpha$ of pages versus entropy efficiency, where non-power-law pages are marked as red.}
    \label{fig:alpha}
    \vspace{-0.0cm}
\end{figure}
Figure \ref{fig:alpha} provides the scatterplot between each page's $\alpha$ and its entropy efficiency. 

In conclusion, we claim that (1) there are positive correlations and reinforcement among the number of editors, the efficiency of edit distributions among editors, and the readership of pages, (2) although the total energy of a page does not correlates with the quality/readership of the page, it  clearly identifies the group of bad pages (i.e., low energy ones), and (3) there is a phase transition in the entropy measurements with the number of editors.

\section{V. Conclusion} We have studied the efficiency of social collaborative behaviors through thermodynamic principles; in particular, we discovered 
(1) editors' energy levels correspond to the logarithmic of their number of edits, and the power-law of edit distributions arises naturally from thermodynamic principles, and
(2) while entropy or entropy reduction characterizes order, maximizing entropy efficiency is one of the basic thermodynamic principles. By applying these measurements to the Wikipedia dataset, we see that (1) Wikipedia is becoming more efficient, (2) entropy efficiency is correlated with the quality of the social collaboration, and (3) there is a suggestive phase transition separated by a particular number of editors, the system may self-organize into efficient and ordered states if the threshold is passed. Note that although we have used ``number of edits" as the source of contributions, such analysis is also applicable to other metrics, e.g., length of contributions from editors.

In the future work, we like to understand what causes the phase transition by developing both microscopic and macroscopic  evolutionary models of editing behaviors. Such models may give insights and predictions for the success and failure of social collaborations.

\section{Acknowledgment}
We like to thank Dr. Bongwon Suh for extracting the Wikipedia meta-data for this analysis, and anonymous reviewers and Dr. Todd Hylton for valuable comments. This work was partly sponsored by Defense Advanced Research Projects Agency (DARPA) Defense Sciences Office (DSO) Program: Thermodynamically Evolving Robust and Adaptive Physical Intelligence (THERA-PI) under Contract No. HR0011-10-C-052.

\bibliography{reference}  

\appendix
\begin{theorem}
Given a collection $\{v_{i}|i=1..N\}$ satisfying power-law distribution with power-law coefficient $\alpha$, let average energy be ${E}=\frac{\sum_{1}^{N} \log(v_{i})}{N}$. If the minimum value of $v$ is 1, ${E}=\frac{1}{\alpha-1}$.
\end{theorem}
\begin{proof}
According to \cite{Newman_2005}, we have 
\(
\alpha = 1+\frac{N}{\sum_{i=1}^{N} \log(\frac{v_i}{v_{\min}})}
\)
Since $v_{\min}=1$, we have
\(
\alpha = 1+\frac{1}{E}, \mbox{i.e.}, {E}  = \frac{1}{\alpha-1}
\)
\end{proof}

\end{document}